\documentclass[conference]{IEEEtran}

\usepackage{amsmath,amssymb,amsthm,mathtools}
\usepackage{bm,bbm}
\usepackage{graphicx}
\usepackage{booktabs}
\usepackage{xcolor}
\let\labelindent\relax
\usepackage{enumitem}
\usepackage[hidelinks]{hyperref}
\usepackage[capitalize,noabbrev]{cleveref}
\usepackage{microtype}

\newtheorem{proposition}{Proposition}

\theoremstyle{definition}

\newtheorem{remark}{Remark}
\newtheorem{example}{Example}

\newcommand{\R}{\mathbb{R}}
\newcommand{\C}{\mathbb{C}}
\newcommand{\E}{\mathbb{E}}
\newcommand{\herm}{\mathsf{H}}
\newcommand{\CN}{\mathcal{CN}}
\newcommand{\Pa}{\mathrm{Pa}}
\DeclareMathOperator{\tr}{tr}

\newcommand{\MI}{I}
\newcommand{\param}{\bm{\eta}}
\newcommand{\rateset}{\mathcal{R}^{\mathrm{ach}}}
\newcommand{\indset}{\mathcal{S}}
\newcommand{\feasset}{\mathcal{F}}
\newcommand{\sigtau}{\sigma_{\tau}}
\newcommand{\rhotau}{\widehat{\rho}_{\tau}}

\title{Differentiable Conditional Mutual Information for\\
Multi-Terminal Linear Gaussian Wireless Networks}

\author{%
\IEEEauthorblockN{Tadashi~Wadayama and Siqi~Na}\\
\IEEEauthorblockA{Department of Computer Science,
Nagoya Institute of Technology, Nagoya 466-8555, Japan\\
Email: wadayama@nitech.ac.jp}%
}

\begin{document}
\maketitle

\begin{abstract}
The rate regions of multi-terminal Gaussian channels
(multiple-access, broadcast, interference, relay) are delimited by
conditional mutual informations
$I(V_A;V_B\,|\,V_C)$ among groups of input and output
nodes; bringing such channels under differentiable physical-layer
design therefore hinges on evaluating any such conditional MI, and
its gradient, on a unified computation graph. Modeling the network
as a linear Gaussian directed acyclic graph (Gaussian-DAG), we
obtain $I(V_A;V_B\,|\,V_C)$ in closed form: from the
node-pair covariances produced by one K-recursion forward pass, it
is a log-determinant difference of two sub-block Schur complements
of the support covariance. The construction is built entirely from
automatic-differentiation (AD) primitives, so any differentiable
function of finitely many conditional MIs is end-to-end
differentiable in the design parameters; this broad class includes
linear objectives (weighted sum-rate, secrecy), the rate
functions of standard multi-terminal rate regions, and non-linear
composites of these. A single reverse-mode AD sweep yields the
Wirtinger gradient with respect to all controllable factors at once,
so any such objective can be handled by projected gradient
iterations without problem-specific gradient derivation.
We demonstrate the framework on three
experiments: rate-region maximization for a two-user MIMO
multiple-access channel, secure precoding on a MIMO wiretap channel
(extended to a Lagrangian sweep that traces the empirical
leakage-rate Pareto curve, with the classical secrecy rate as the
$\lambda\!=\!1$ point), and the same rate-region objective applied
to a larger multi-hop multiple-access network.
\end{abstract}

\begin{IEEEkeywords}
Conditional mutual information, linear Gaussian DAG, K-recursion,
Schur complement, multi-terminal rate region, differentiable
programming, automatic differentiation, projected gradient method.
\end{IEEEkeywords}

\section{Introduction}
\label{sec:intro}


Mutual information (MI) has long served as a central design objective
for vector Gaussian communication channels: from Telatar's waterfilling
solution for the single-link MIMO capacity~\cite{telatar1999} to the
Palomar--Verd\'u closed-form gradient of MI with respect to the
transmit precoder~\cite{palomar2006}, which underpins a wide class of
MIMO precoding, beamforming, and sensing designs. A common feature of
these classical results is that each new topology (a multi-hop
amplify-and-forward relay, a diamond network, a precoded broadcast, a
cell-free MIMO front-end) demands its own analytical derivation of
$\nabla_{\bm{F}}\MI$, with no topology-agnostic mechanism for these
gradient computations. This
is increasingly at odds with how modern wireless systems are designed:
the physical layer is a composition of cascaded linear-Gaussian stages
whose parameters one would like to tune jointly, end to end, by
gradient descent in a single computation graph, a
\emph{differentiable-programming} workflow analogous to that used in
deep learning.

A natural way to bring MI-based design under this umbrella is the
\emph{linear Gaussian directed acyclic graph} (Gaussian-DAG)
framework of~\cite{wadayama2026arxiv}: every node is a circular
complex Gaussian vector and every edge a linear transformation, and
the node-pair covariances $\bm{K}_{jk}\!=\!\E[V_jV_k^{\herm}]$ are
computed by a topological-order recursion (the \emph{K-recursion})
using only matrix products, sums, and Hermitian transposes, the
standard primitives of complex automatic-differentiation (AD) engines.
The local-to-global covariance map underlying the
K-recursion is itself classical~\cite{shachter1989,geiger1994,sullivant2010};
\cite{wadayama2026arxiv} contributes its formulation as a single
differentiable forward operator on which AD-based optimization can be
built. In that framework, the end-to-end log-determinant MI between
a single input node and a single output node is a differentiable
scalar functional of the controllable factors, and reverse-mode AD
delivers exact Wirtinger gradients for arbitrary topologies in a
single backward pass.

Multi-terminal design objectives, however, are not of this single,
unconditional form: the two-user MAC pentagon is bounded by the
conditional MIs $\MI(V_{X_1};V_Y\,|\,V_{X_2})$,
$\MI(V_{X_2};V_Y\,|\,V_{X_1})$, and $\MI(V_{X_1},V_{X_2};V_Y)$; a
secrecy rate is the MI difference $\MI(V_X;V_Y)\!-\!\MI(V_X;V_Z)$; broadcast, interference, and relay rate
functions are signed combinations of such conditional MIs. Bringing
multi-terminal design under differentiable programming therefore
hinges on evaluating \emph{any} conditional mutual information
$\MI(V_A;V_B\,|\,V_C)$ on the DAG through one common,
differentiable computation graph.

For context, three alternative routes to differentiable MI-based
design merit brief comment. Variational and contrastive estimators
(MINE~\cite{mine2018}, InfoNCE~\cite{infonce2018}) approximate MI
through a learned lower bound; their gradients are stochastic,
sample-based quantities tied to that bound. A score-function-based
DAG framework~\cite{wadayama2026dag} avoids that variational gap, in
return for fitting a denoising score network and integrating via
Monte-Carlo. In the linear-Gaussian setting these costs are
eliminated: covariances propagate analytically along the DAG via the
K-recursion of~\cite{wadayama2026arxiv}, and MI gradients on the
design parameters are recovered by a single reverse-mode pass with
neither training nor sampling.

\subsubsection*{Contribution}
This paper extends the differentiable Gaussian-DAG approach from a
single MI to arbitrary conditional mutual information, through a
three-step computational chain.
\emph{(i)}~From the K-blocks produced by one K-recursion forward pass,
the conditional mutual information
$\MI(V_A;V_B\,|\,V_C)$ of \emph{any} disjoint node
subsets $A,B,C$ is obtained in closed form, as a log-determinant
difference of sub-block Schur complements of the support covariance
(\cref{prop:cmi_schur}), a procedure built entirely from
AD-compatible primitives.
\emph{(ii)}~A \emph{conditional-MI objective} is any function of
finitely many conditional mutual informations on the DAG that is
differentiable in those arguments; by the chain rule, every such
objective is end-to-end differentiable in the design parameters.
Familiar instances include weighted sum-rate, fairness, and secrecy
objectives, as well as the rate functions of multi-terminal rate
regions (the multiple-access pentagon, the Gaussian broadcast region
with dirty-paper coding, the Han--Kobayashi inner bound, and
decode-/compress-and-forward relay bounds, once their coding order,
auxiliary-variable structure, and time-sharing choice are fixed).
Non-linear composites, such as the composite sigmoid surrogate of
rate-region outage, also fall within the class.
\emph{(iii)}~A single reverse-mode AD sweep through the shared graph
yields the Wirtinger gradient with respect to all controllable factors
at once, so any such objective is amenable to projected-gradient
optimization with no problem-specific derivation.
The stochastic optimization of the outage surrogate under fading is
left to a companion paper.

\subsubsection*{Paper Organization}
\Cref{sec:preliminaries} reviews the linear Gaussian DAG model and the
K-recursion of~\cite{wadayama2026arxiv}. \Cref{sec:cmi} establishes
the closed-form, differentiable evaluator of arbitrary conditional
mutual information (\cref{prop:cmi_schur}) and its reverse-mode
Wirtinger gradient. \Cref{sec:objectives} introduces conditional-MI
objectives (from linear combinations to multi-terminal rate-region
functions) and their optimization by the projected gradient method.
\Cref{sec:experiments} reports numerical results,
and \cref{sec:conclusion} concludes.

\section{Preliminaries}
\label{sec:preliminaries}

This section reviews the linear Gaussian DAG model and the K-recursion
of~\cite{wadayama2026arxiv}, fixing notation and the K-block primitives
used throughout the paper.

\subsection{Notation}
\label{sec:notation}
Uppercase italic letters (e.g., $X, V_j$) denote random variables or
random vectors, while boldface letters (e.g., $\bm{A}, \bm{\Sigma}$)
denote deterministic vectors and matrices.
$\bm{A}^{\herm}$ denotes the Hermitian (conjugate) transpose;
$\bm{A}^{*}$ denotes the entry-wise complex conjugate;
$\bm{\Sigma}\!\succ\!\bm{0}$ means Hermitian positive definite (HPD);
$\bm{\Sigma}\!\succeq\!\bm{0}$ means Hermitian positive semidefinite
(PSD); $\bm{I}_d$ is the $d\!\times\!d$ identity matrix;
$\tr(\cdot)$, $\det(\cdot)$, $\|\cdot\|_F$ are trace, determinant,
and Frobenius norm; $\E[\cdot]$ and $\Pr[\cdot]$ denote expectation
and probability; $Y\!\sim\!\CN(\bm{\mu},\bm{\Sigma})$ denotes a
circular complex Gaussian random vector with mean $\bm{\mu}$ and
covariance $\bm{\Sigma}$; and
$\sigtau(x)\!\triangleq\!(1+e^{-x/\tau})^{-1}$ is the sigmoid with
temperature $\tau\!>\!0$.

\subsection{Linear Gaussian DAG Model and K-Recursion}
\label{sec:krec_review}
Let $\mathcal{G}\!=\!(\mathcal{V},\mathcal{E})$ be a topologically
ordered DAG with node index set $\mathcal{V}\!=\!\{1,\dots,M\}$.
We adopt the convention that the first $K$ nodes are the user-input
roots, $\mathcal{R}\!=\!\{1,\dots,K\}\!\subsetneq\!\mathcal{V}$, with
$K\!<\!M$ and one root per rate-bearing user, so that node
$k\!\in\!\mathcal{R}$ carries the input signal of the rate-$R_k$
user, and the rate-index set $[K]\!=\!\{1,\dots,K\}$ coincides with
the root index set $\mathcal{R}$. Each node $j\!\in\!\mathcal{V}$ carries a complex random vector
$V_j\!\in\!\C^{d_j}$. The roots
$\{V_r\}_{r\in\mathcal{R}}$ are jointly circular complex
Gaussian with prescribed block covariance
$\bm{\Sigma}_{\mathcal{R}}\!\triangleq\![\bm{\Sigma}_{r,r'}]_{r,r'\in\mathcal{R}}\!\succeq\!\bm{0}$:
the diagonal blocks $\bm{\Sigma}_{r,r}\!=\!\bm{\Sigma}_r$ are the
\emph{input covariances} of the rate-$R_r$ users, and the
off-diagonal blocks $\bm{\Sigma}_{r,r'}$ ($r\!\neq\!r'$) encode any
prescribed inter-source correlation. The independent-input case
$\bm{\Sigma}_{r,r'}\!=\!\bm{0}$ ($r\!\neq\!r'$) is adopted
in~\cref{sec:objectives,sec:experiments}; correlated roots are
admissible for distributed source coding problems such as
Slepian--Wolf and CEO. For each non-root node
$j\!\in\!\mathcal{V}\!\setminus\!\mathcal{R}$, let
$\Pa(j)\!\subseteq\!\mathcal{V}$ denote the set of in-neighbors of
$j$ in $\mathcal{G}$ (its parent nodes); the node obeys the
structural equation
\begin{equation}
V_j\!=\!\sum_{i\in\Pa(j)}\bm{A}_{ji}V_i+Z_j,\quad
Z_j\!\sim\!\CN(\bm{0},\bm{\Sigma}_j),
\label{eq:struct_eq}
\end{equation}
with edge matrices $\bm{A}_{ji}\!\in\!\C^{d_j\times d_i}$ and
additive channel noise $Z_j$. The noises
$\{Z_j\}_{j\in\mathcal{V}\setminus\mathcal{R}}$ are mutually
independent and independent of the user inputs
$\{V_r\}_{r\in\mathcal{R}}$.

\subsubsection{Edge factorization and design parameter}
Each edge matrix admits a multiplicative decomposition
\begin{equation}
\bm{A}_{ji}\!=\!\bm{A}_{ji}^{(1)}\bm{A}_{ji}^{(2)}\cdots\bm{A}_{ji}^{(L_{ji})}
\label{eq:edge_factor}
\end{equation}
into $L_{ji}\!\geq\!1$ matrix factors, each labeled either
\emph{controllable} (a tunable design variable: a precoder, relay
matrix, beamformer, etc.) or \emph{constant} (a fixed physical
quantity such as a channel realization $\bm{H}$). Indexing the
controllable factors by a designated set
$\mathcal{D}\!\subseteq\!\{(j,i,\ell)\}$ yields the \emph{design
parameter}
\begin{equation}
\param\!\triangleq\!\{\bm{A}_{ji}^{(\ell)}:(j,i,\ell)\!\in\!\mathcal{D}\},
\label{eq:design_param}
\end{equation}
to be optimized over a \emph{feasible set} $\feasset$ that collects
the power and structural constraints of the design: per-user
Frobenius power balls $\|\bm{F}_k\|_F^2\!\le\!P_k$, total-power balls,
per-antenna power, unitary precoders, or hybrid analog/digital
structures. Both
the K-recursion below and the conditional-MI evaluation
of~\cref{sec:logdet_mi} treat the constant factors as fixed and
expose $\param$ as the optimization variable on which Wirtinger
gradients are taken in~\cref{sec:wirtinger}.

\subsubsection{Shared controllable factors}
It is common for a single tunable matrix to enter many edges at once: a
relay node $i$'s processing matrix $\bm{F}_i$, for instance, is applied to
its received signal and then propagated outward, so every outgoing edge
$(j,i)$ with $i\!\in\!\Pa(j)$ carries the factorization
$\bm{A}_{ji}\!=\!\bm{H}_{ji}\bm{F}_i$ with $\bm{H}_{ji}$ the fixed channel
realization. The design parameter $\param$ therefore aliases $\bm{F}_i$
across multiple index positions, and the AD computation graph realizes
those positions as fan-out edges from a single leaf node for $\bm{F}_i$.
The backward pass aggregates the gradient contributions arriving along all
of these fan-out edges into one $\nabla_{\bm{F}_i^{*}} U$ via the chain
rule, at the same asymptotic cost as the forward
K-recursion~\cite{baydin2018ad}. \Cref{sec:exp_random_mac} exploits this
mechanism to jointly tune the nine relay processing matrices of a
multi-hop MAC under a single shared power budget.

\subsubsection{The K-recursion}
The local-to-global covariance map of a linear
Gaussian DAG, from the local edge matrices and noise covariances to
the joint node-pair covariances, is
classical~\cite{shachter1989,geiger1994,sullivant2010} and equivalent
to the standard closed form $(\bm{I}\!-\!\bm{A})^{-1}\bm{\Sigma}(\bm{I}\!-\!\bm{A})^{-\herm}$;
\cite{wadayama2026arxiv} cast it as a single \emph{differentiable}
forward operator, the \emph{K-recursion}, which constructs every
node-pair covariance
$\bm{K}_{jk}\!\triangleq\!\E[V_jV_k^{\herm}]$
in topological order: for $j,k\!\in\!\mathcal{V}$ with $j\!\ge\!k$,
\begin{equation}
\bm{K}_{jk}\!=\!
\begin{cases}
\bm{\Sigma}_{j,k}, & j,k\!\in\!\mathcal{R},\\[1pt]
\sum_{i\in\Pa(j)}\!\bm{A}_{ji}\bm{K}_{ik}, & j\!\notin\!\mathcal{R},\ k\!<\!j,\\[1pt]
\sum_{i,i'\in\Pa(j)}\!\bm{A}_{ji}\bm{K}_{ii'}\bm{A}_{ji'}^{\herm}\!+\!\bm{\Sigma}_j, & j\!\notin\!\mathcal{R},\ k\!=\!j,
\end{cases}
\label{eq:krec}
\end{equation}
with the Hermitian-flip convention
$\bm{K}_{ab}\!=\!\bm{K}_{ba}^{\herm}$ for $a\!<\!b$. The first case
reads the diagonal and inter-source blocks of
$\bm{\Sigma}_{\mathcal{R}}$ above directly as the base case. In
particular, the cross-block $\bm{K}_{ik}$ in the second line
of~\eqref{eq:krec} is read using the Hermitian-flip convention
whenever $i\!<\!k$, and likewise $\bm{K}_{ii'}$ in the third line
whenever $i\!<\!i'$. Each step
uses
only matrix products, sums, and Hermitian transposes, so the full
collection $\{\bm{K}_{jk}\}_{j\geq k}$ is obtained in a single forward
pass and is a smooth function of the design parameter $\param$.
This inverse-free, single-pass, all-pairs form is
what makes the K-recursion directly composable with reverse-mode AD;
the conditional-MI evaluator of~\cref{sec:cmi} inherits that
composability.

\begin{remark}[Justification of~\eqref{eq:krec}]
\label{rem:krec_proof}
The base case $\bm{K}_{r,r'}\!=\!\bm{\Sigma}_{r,r'}$ is the
definition of the joint root covariance $\bm{\Sigma}_{\mathcal{R}}$
and requires no independence among the roots. For a
non-root $j$ and $k\!<\!j$, substituting~\eqref{eq:struct_eq} into the
defining expectation of $\bm{K}_{jk}$ and using that $Z_j$ is
independent of every $V_k$ with $k\!<\!j$ collapses the cross
terms and leaves the cross-block formula. The self-block case
$k\!=\!j$ arises from substituting~\eqref{eq:struct_eq} on both
factors, which expands into a double sum over $\Pa(j)\!\times\!\Pa(j)$
plus the noise auto-covariance $\bm{\Sigma}_j$. Topological ordering
guarantees that every right-hand side in~\eqref{eq:krec} refers to
index pairs strictly preceding $(j,k)$, so one forward sweep through
$\mathcal{V}$ computes the entire upper-triangular collection.
See~\cite{wadayama2026arxiv} for the full argument.
\end{remark}

Tracking every cross-block, not merely the diagonal
$\{\bm{K}_{jj}\}_j$, is essential: whenever a node $j$ has two or more
parents $i\!\neq\!i'\!\in\!\Pa(j)$, the self-block $\bm{K}_{jj}$
in~\eqref{eq:krec} contains terms
$\bm{A}_{ji}\bm{K}_{ii'}\bm{A}_{ji'}^{\herm}$ involving the
\emph{cross-block} $\bm{K}_{ii'}$ between two distinct upstream nodes,
which in turn requires cross-blocks further upstream. This dependence
on parent cross-covariances is unavoidable in any DAG with merging
nodes.

\begin{remark}[Computational cost]
\label{rem:complexity}
Let $d_{\max}\!\triangleq\!\max_j d_j$. For a sparse DAG with
$|\Pa(j)|\!=\!O(1)$, the recursion~\eqref{eq:krec} produces the
$O(M^2)$ blocks $\{\bm{K}_{jk}\}_{j\geq k}$ in $O(M^2 d_{\max}^3)$
arithmetic operations and $O(M^2 d_{\max}^2)$ memory. This is a
worst-case uniform-dimension estimate; variable node dimensions yield
the corresponding blockwise sum. For the experiments
of~\cref{sec:experiments} ($M\!\leq\!12$, $d\!=\!4$), the resulting
cost is negligible in PyTorch double precision. The
reverse-mode AD backward pass runs in the same asymptotic time; the
only added cost is keeping the forward activations in memory for the
backward pass.
\end{remark}

\section{Differentiable Conditional Mutual Information}
\label{sec:cmi}

This section is the analytic core of the paper. We use the K-blocks
of~\cref{sec:krec_review} to obtain
$\MI(V_A;V_B\,|\,V_C)$ in closed form
(\cref{prop:cmi_schur}) and show that, via the Wirtinger calculus, it
is end-to-end differentiable in the design parameter $\param$ with
gradients delivered by a single reverse-mode AD pass.

\subsection{Log-Det Conditional MI via Block Extraction}
\label{sec:logdet_mi}
We now use the K-blocks of \cref{sec:krec_review} to evaluate
arbitrary conditional mutual information on $\mathcal{G}$.
Let $A,B,C$ be \emph{disjoint} subsets of $\mathcal{V}$ whose
conditional mutual information we wish to evaluate. Choose an
arbitrary but fixed ordering of their union and set
\begin{equation}
S\!\triangleq\!A\!\cup\!B\!\cup\!C
\!=\!\{s_1,\dots,s_n\}
\label{eq:support}
\end{equation}
(the \emph{support} of the computation), and for any subset
$X\!\subseteq\!\mathcal{V}$ enumerate its elements as
$\{x_1,\dots,x_{|X|}\}$ in a fixed order and
define the stacked vector
$V_X\!\triangleq\![V_{x_1}^{\top}\,\cdots\,V_{x_{|X|}}^{\top}]^{\top}$.
All conditional MI values below are invariant under such orderings
(each amounts to a simultaneous row/column permutation of
$\bm{\Sigma}_{S,S}$).
We adopt two subscript conventions throughout: (i) juxtaposition of
set symbols denotes set union, $XY\!\triangleq\!X\!\cup\!Y$; (ii) a comma
in a covariance subscript separates the row index set from the column
index set,
$\bm{\Sigma}_{X,Y}\!\triangleq\!\E[V_XV_Y^{\herm}]\!=\!
[\bm{K}_{x_iy_j}]_{x_i\in X,\,y_j\in Y}$.
Thus $BC$ below denotes the union $B\!\cup\!C$, not a Cartesian or
matrix product, and $\bm{\Sigma}_{A,BC}$ is the cross-covariance
block between $V_A$ and $V_{B\cup C}$.

In particular, the support stacked vector $V_S$ has the block
covariance
\begin{equation}
\bm{\Sigma}_{S,S}
\!\triangleq\!
\begin{bmatrix}
\bm{K}_{s_1 s_1} & \!\!\cdots\!\! & \bm{K}_{s_1 s_n}\\
\vdots & \!\!\ddots\!\! & \vdots\\
\bm{K}_{s_n s_1} & \!\!\cdots\!\! & \bm{K}_{s_n s_n}
\end{bmatrix}
\!=\![\bm{K}_{s_is_j}]_{1\le i,j\le n},
\label{eq:block_extract}
\end{equation}
formed by stacking only the $n^2$ K-blocks indexed by $S$, using the
Hermitian-flip convention of~\eqref{eq:krec}
($\bm{K}_{ab}\!=\!\bm{K}_{ba}^{\herm}$ for $a\!<\!b$) for any
sub-diagonal entry $\bm{K}_{s_is_j}$ with $s_i\!<\!s_j$. We refer
to this assembly as \emph{block extraction}; every row/column
sub-block
\[
\bm{\Sigma}_{A,A},\;\bm{\Sigma}_{A,C},\;\bm{\Sigma}_{C,C},\;
\bm{\Sigma}_{A,BC},\;\bm{\Sigma}_{BC,BC},\;\ldots
\]
entering the Schur complements below is read off
from~\eqref{eq:block_extract} by selecting the corresponding row and
column index ranges.

Recall that for a positive-definite, block-partitioned Hermitian
matrix
$\bm{M}\!=\!\bigl[\begin{smallmatrix}\bm{P}&\bm{N}\\\bm{N}^{\herm}&\bm{Q}\end{smallmatrix}\bigr]$
with $\bm{Q}\!\succ\!\bm{0}$, the \emph{Schur complement} of
$\bm{Q}$ in $\bm{M}$ is
$\bm{P}\!-\!\bm{N}\bm{Q}^{-1}\bm{N}^{\herm}$; when $\bm{M}$ is the
covariance of a jointly Gaussian pair $(V_P,V_Q)$, this
Schur complement is precisely the conditional covariance of
$V_P$ given $V_Q$~\cite[Ch.~8]{cover2006}. We call these
\emph{sub-block Schur complements of} $\bm{\Sigma}_{S,S}$ because
they are formed from selected row and column blocks of the support
covariance.

The following proposition turns the K-blocks produced
by~\eqref{eq:krec} into a closed-form, log-determinant expression
for $\MI(V_A;V_B\,|\,V_C)$ on \emph{any} disjoint
node subsets $A,B,C$ of $\mathcal{G}$. The
analytic statement is made in the positive-definite regime
$\bm{\Sigma}_{S,S}\!\succ\!\bm{0}$; a simple sufficient condition is
that the joint root covariance $\bm{\Sigma}_{\mathcal{R}}$ and every
non-root noise covariance $\bm{\Sigma}_j$ are positive definite, as
in the common full-rank input/noise setting of wireless channels.
Rank-deficient cases (which may arise, for instance, when
$\bm{\Sigma}_{\mathcal{R}}$ is itself rank-deficient) lie outside the
exact log-determinant statement below, and near-singular instances
are treated in numerical implementations by adding a small positive
multiple of the identity to the conditioning matrix before Cholesky
factorization (see~\cref{rem:cholesky}).

\begin{proposition}[Closed-form conditional MI on the K-recursion]
\label{prop:cmi_schur}
Let $\mathcal{G}$ be a linear Gaussian DAG with K-blocks
$\{\bm{K}_{jk}\}$ produced by~\eqref{eq:krec}, fix disjoint subsets
$A,B,C\!\subseteq\!\mathcal{V}$ with $A,B$ non-empty (the
conditioning set $C$ may be empty), set
$S\!=\!A\!\cup\!B\!\cup\!C$ and the block covariance
$\bm{\Sigma}_{S,S}$ of~\eqref{eq:block_extract}, and adopt the
convention
$\bm{\Sigma}_{A|\emptyset}\!\triangleq\!\bm{\Sigma}_{A,A}$. Suppose the support covariance is positive definite,
$\bm{\Sigma}_{S,S}\!\succ\!\bm{0}$. Then:
\begin{enumerate}[label=\textup{(\alph*)},leftmargin=*]
\item For $C\!\ne\!\emptyset$, the conditional covariances of $V_A$
given $V_C$ and given $V_{BC}$ are sub-block Schur complements
of $\bm{\Sigma}_{S,S}$,
\begin{equation}
\begin{aligned}
\bm{\Sigma}_{A|C}
 &\!=\!\bm{\Sigma}_{A,A}\!-\!\bm{\Sigma}_{A,C}\bm{\Sigma}_{C,C}^{-1}\bm{\Sigma}_{C,A},\\
\bm{\Sigma}_{A|BC}
 &\!=\!\bm{\Sigma}_{A,A}\!-\!\bm{\Sigma}_{A,BC}\bm{\Sigma}_{BC,BC}^{-1}\bm{\Sigma}_{BC,A};
\end{aligned}
\label{eq:schur}
\end{equation}
for $C\!=\!\emptyset$, the first identity reduces to the convention
$\bm{\Sigma}_{A|\emptyset}\!\triangleq\!\bm{\Sigma}_{A,A}$ stated above,
and the second holds with $BC\!=\!B$ (which is nonempty since $B$ is
nonempty).
\item The conditional mutual information of $(V_A,V_B)$
given $V_C$ admits the log-determinant closed form
\begin{equation}
\MI(V_A;V_B\,|\,V_C)
\!=\!\log\det\bm{\Sigma}_{A|C}\!-\!\log\det\bm{\Sigma}_{A|BC};
\label{eq:cmi_logdet}
\end{equation}
the unconditional case $C\!=\!\emptyset$ specializes to
$\MI(V_A;V_B)\!=\!\log\det\bm{\Sigma}_{A,A}\!-\!\log\det\bm{\Sigma}_{A|B}$.
\end{enumerate}
\end{proposition}

\begin{proof}
\emph{(a)} The linear Gaussian DAG model
of~\cref{sec:krec_review} makes every $V_j$ an affine function
of the jointly Gaussian roots $\{V_r\}_{r\in\mathcal{R}}$ and
of the mutually independent Gaussian noises $\{Z_j\}$, so
$V_S$ is circular complex Gaussian with covariance
$\bm{\Sigma}_{S,S}\!\succ\!\bm{0}$. Every nonempty principal
sub-block of $\bm{\Sigma}_{S,S}$ is therefore positive definite,
and the same holds for every Schur complement obtained from it.
For nonempty disjoint $A,X\!\subseteq\!S$, the joint vector
$(V_A,V_X)$ has covariance
$\bm{\Sigma}_{AX,AX}\!\succ\!\bm{0}$, and the Schur-complement /
conditional-covariance correspondence recalled above gives
$\bm{\Sigma}_{A|X}\!=\!\bm{\Sigma}_{A,A}\!-\!\bm{\Sigma}_{A,X}\bm{\Sigma}_{X,X}^{-1}\bm{\Sigma}_{X,A}\!\succ\!\bm{0}$
as the conditional covariance of the (circular complex Gaussian)
$V_A\,|\,V_X$. Specializing to $X\!=\!C$ (when
$C\!\neq\!\emptyset$) and $X\!=\!BC$ yields the two formulas
in~\eqref{eq:schur}; the boundary case $C\!=\!\emptyset$ is covered
by the convention
$\bm{\Sigma}_{A|\emptyset}\!\triangleq\!\bm{\Sigma}_{A,A}\!\succ\!\bm{0}$.

\emph{(b)} Since $\bm{\Sigma}_{S,S}\!\succ\!\bm{0}$, the conditional
covariances obtained in part~(a) are positive definite, so the
circular-complex-Gaussian differential entropy formula applies to
both. By the entropy chain rule,
$\MI(V_A;V_B\,|\,V_C)\!=\!h(V_A\,|\,V_C)\!-\!h(V_A\,|\,V_{BC})$,
where $h(V_A\,|\,V_\emptyset)\!\triangleq\!h(V_A)$
for $C\!=\!\emptyset$. A circular complex Gaussian vector with
covariance $\bm{\Sigma}\!\succ\!\bm{0}$ has differential entropy
$h(\cdot)\!=\!\log\det(\pi e\,\bm{\Sigma})$ in nats; substituting
the conditional covariances of~(a) and noting that the
$\log\det(\pi e\,\bm{I})$ terms cancel
yields~\eqref{eq:cmi_logdet}.
\end{proof}

The expression~\eqref{eq:cmi_logdet} is asymmetric only in its
computational form: exchanging $A$ and $B$ throughout yields the
equivalent formula
$\MI(V_A;V_B\,|\,V_C)\!=\!\log\det\bm{\Sigma}_{B|C}\!-\!\log\det\bm{\Sigma}_{B|AC}$
based on the conditional covariance of $V_B$, with the same
value by the $A\!\leftrightarrow\!B$ symmetry of conditional MI.

\Cref{fig:schur} visualizes the two sub-block Schur complements
of~\eqref{eq:schur} as selections from the $3\!\times\!3$ block
partition of $\bm{\Sigma}_{S,S}$ by the disjoint subsets $A,B,C$.
In each panel the anchor cell $\bm{\Sigma}_{A,A}$ is shaded green,
the blocks entering the correction term are shaded blue, and the
uninvolved blocks are greyed out. Panel~(a) shows $\bm{\Sigma}_{A|C}$,
which uses only the $A$- and $C$-blocks; panel~(b) shows
$\bm{\Sigma}_{A|BC}$, which uses every block of $\bm{\Sigma}_{S,S}$.

\begin{figure}[t]
\centering
\includegraphics[width=\linewidth]{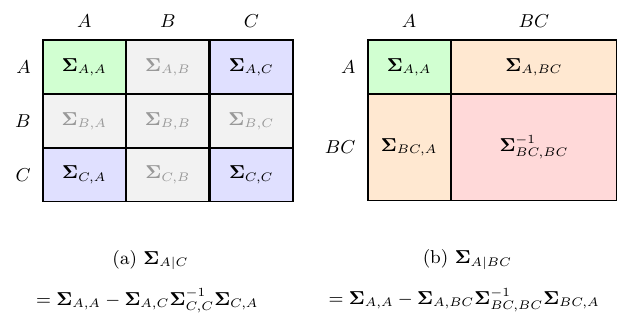}
\caption{Sub-block Schur complements formed from the support
covariance $\bm{\Sigma}_{S,S}$: (a) $\bm{\Sigma}_{A|C}$;
(b) $\bm{\Sigma}_{A|BC}$.}
\label{fig:schur}
\end{figure}

\begin{remark}[Numerical implementation]
\label{rem:cholesky}
The exact formulas of \cref{prop:cmi_schur} are stated under positive
definiteness; in finite-precision arithmetic, we evaluate a slightly
regularized version by adding a small diagonal jitter at two distinct
points:
\emph{(i)} before the Cholesky factorization of the \emph{conditioning
matrices} $\bm{\Sigma}_{C,C},\bm{\Sigma}_{BC,BC}$ in the
Schur-complement solves of~\eqref{eq:schur}; and
\emph{(ii)} before the Cholesky factorization of the resulting
\emph{conditional covariances} $\bm{\Sigma}_{A|C},\bm{\Sigma}_{A|BC}$
for the log-determinants of~\eqref{eq:cmi_logdet}.
The inverses $\bm{\Sigma}_{C,C}^{-1}$ and $\bm{\Sigma}_{BC,BC}^{-1}$
in~\eqref{eq:schur} appear for notational clarity; they are never
formed explicitly in practice. Each sub-block Schur complement is
evaluated through a Cholesky-based linear solve, and each
log-determinant in~\eqref{eq:cmi_logdet} through the sum of the
diagonal logs of the corresponding Cholesky factor: a real scalar,
since the diagonal of the Cholesky factor of a Hermitian
positive-definite matrix is positive real, and any small imaginary
residual produced by complex-tensor floating-point arithmetic is
discarded by construction. At point~(i),
when fading realizations or rank-deficient controllable factors drive
the conditioning matrix close to singular, the diagonal jitter
stabilizes the solve. At point~(ii), accumulated floating-point error
along the K-recursion / block-extraction / Schur pipeline may also
cause the output covariances $\bm{\Sigma}_{A|C},\bm{\Sigma}_{A|BC}$ to
drift away from Hermitian positive-definite form; we therefore
symmetrize each as
$\tfrac{1}{2}(\bm{\Sigma}\!+\!\bm{\Sigma}^{\herm})$ and apply a
matching diagonal jitter before its log-determinant is evaluated.
\end{remark}

On the positive-definite domain of~\cref{prop:cmi_schur}, evaluating
any $\MI(V_A;V_B\,|\,V_C)$ on $\mathcal{G}$ thus
reduces to three smooth ingredients applied to the K-blocks indexed
by the support $S$ of~\eqref{eq:support}:
(i) block extraction~\eqref{eq:block_extract} of $\bm{\Sigma}_{S,S}$,
(ii) two sub-block Schur complements~\eqref{eq:schur},
and (iii) two log-determinants in~\eqref{eq:cmi_logdet}.
With a fixed diagonal regularization $\bm{\Sigma}\!\mapsto\!\bm{\Sigma}\!+\!\epsilon\bm{I}$,
the same pipeline evaluates a smooth regularized surrogate on the
entire feasible set $\feasset$. A conditional-differential-entropy
specialization of the same pipeline (one Schur complement and one
log-determinant) is recorded in Appendix~\ref{app:entropy}.

\subsection{Wirtinger Calculus and Reverse-Mode AD}
\label{sec:wirtinger}
For a real-valued function $f(\bm{\Theta})\!\in\!\R$ of a complex
matrix variable $\bm{\Theta}\!\in\!\C^{p\times q}$, the Wirtinger
conjugate-side derivative
$\nabla_{\bm{\Theta}^{*}}f\!\triangleq\!(\partial f/\partial\bm{\Theta}^{*})^T$
is the steepest-ascent direction in the standard real-Euclidean
metric~\cite{schreier_scharf2010,palomar2006} and satisfies
$\partial f/\partial\bm{\Theta}\!=\!(\partial f/\partial\bm{\Theta}^{*})^{*}$,
so that one of the two partials suffices to specify $\nabla f$.
When $f$ depends on a tuple of complex matrix variables, such as the
design parameter $\param$ of~\eqref{eq:design_param}, we write
$\nabla_{\param^{*}}f$ for the tuple of conjugate-side derivatives,
one per controllable factor.
Modern automatic-differentiation engines that support complex
tensors (PyTorch~\cite{paszke2019pytorch} in particular) compose
this Wirtinger calculus through every elementary complex primitive
(matrix products, sums, Hermitian transposes, Cholesky
factorizations, triangular solves, and Cholesky-based
log-determinants); a single reverse-mode pass therefore returns
$\nabla_{\bm{\Theta}^{*}}f$ at every complex leaf up to the
convention of the AD engine and finite-precision arithmetic, with no
manual derivation.\footnote{Specific complex-AD conventions are
engine- and version-dependent. In the PyTorch convention used in our
implementation, a complex leaf's \texttt{.grad} attribute is
populated with $2\,\partial f/\partial\bm{\Theta}^{*}$ rather than
$\partial f/\partial\bm{\Theta}^{*}$ itself; this fixed
multiplicative offset is absorbed into the step size of any
first-order optimizer and does not affect optimization behavior.
Readers using a different engine or version should consult its
current documentation.}
Because each step of the K-recursion~\eqref{eq:krec} and each
ingredient of the block-extraction
pipeline~\eqref{eq:block_extract}--\eqref{eq:cmi_logdet} is built
from such primitives, any conditional MI on $\mathcal{G}$ that
satisfies the positive-definiteness condition of~\cref{prop:cmi_schur}
is an end-to-end differentiable function of the design parameter
$\param$ of~\eqref{eq:design_param}, and its Wirtinger gradient
$\nabla_{\param^{*}}\!\MI$ is delivered by a single reverse-mode AD
sweep.

\section{Conditional-MI Objectives and Their Gradient-Based Optimization}
\label{sec:objectives}

\Cref{sec:cmi} evaluates a single conditional mutual information
$\MI(V_A;V_B\,|\,V_C)$ as an end-to-end differentiable
function of the design parameter $\param$. A physical-layer design
problem is seldom a single such quantity; this section shows that
\emph{any} differentiable function of conditional mutual informations,
used as either an objective or a constraint, inherits the same shared
differentiable computation graph (\cref{sec:cmi_graph}) and is
therefore optimized by the projected gradient method. The constant
edge factors (the channel realization) are fixed throughout this
section, so every quantity is written as a function of $\param$ alone.

\subsection{Conditional-MI Objectives}
\label{sec:cmi_obj}

A physical-layer design objective aggregates conditional mutual
informations into a scalar. The most general form is
\begin{equation}
U(\param)\!=\!\Phi\bigl(\MI_1(\param),\dots,\MI_N(\param)\bigr),\quad
\MI_n(\param)\!\triangleq\!\MI(V_{A_n};V_{B_n}\,|\,V_{C_n}),
\label{eq:objective}
\end{equation}
with disjoint $A_n,B_n,C_n\!\subseteq\!\mathcal{V}$ and a
differentiable composition $\Phi\!:\!\R^N\!\to\!\R$ that may further
depend on auxiliary constants (target rates, a temperature, etc.);
we call such a $U$ a \emph{conditional-MI objective}. $\Phi$ need
only be differentiable on an open set containing the attainable
CMI vector, so logarithmic utilities such as proportional fairness
$\log\MI_n$ are handled in their standard regularized form
$\log(\MI_n+\delta)$ in degenerate zero-rate cases. By
\cref{prop:cmi_schur} each $\MI_n$ is a smooth function of $\param$
on the one shared K-recursion graph of~\cref{sec:logdet_mi}, the
chain rule makes $U$ differentiable in $\param$, and a single
reverse-mode AD sweep (\cref{sec:wirtinger}) returns the Wirtinger
gradient $\nabla_{\param^{*}}U$ with respect to \emph{all}
controllable factors at once, independent of the number $N$ of
CMIs requested or of the choice of $\Phi$.

The canonical instance is the linear form
\begin{equation}
U(\param)\!=\!\sum_{n=1}^{N}\alpha_n\,
\MI(V_{A_n};V_{B_n}\,|\,V_{C_n}),
\label{eq:utility}
\end{equation}
with real weights $\{\alpha_n\}$ encoding a weighted sum-rate
criterion (positive weights), a secrecy or leakage penalty
(sign-indefinite weights), or a Lagrangian relaxation of a
conditional-MI constraint (\cref{sec:cmi_constraints}); we call such
a $U$ a \emph{linear conditional-MI objective}. This form covers
many standard Gaussian multi-terminal rate functions: after fixing
the coding order, auxiliary-variable structure, time-sharing
choice, and Gaussian linear parametrization, the rate-function
facets of the multiple-access,
broadcast, interference, and relay rate
regions~\cite{cover2006,elgamal2011} can be written as finite linear
combinations of conditional MIs with weights $\{\alpha_n\}$ of
either sign, so the framework provides a direct, differentiable
handle on the rate-function calculus of multi-terminal information
theory; explicit examples, including the non-linear composite
sigmoid surrogate of rate-region outage and a worked two-user MAC
instance, are collected in~Appendix~\ref{app:rate_region_examples}.

\subsection{Conditional-MI Constraints}
\label{sec:cmi_constraints}

Conditional MIs are equally first-class on the \emph{constraint} side
of a design problem. Combined with a conditional-MI
objective~\eqref{eq:objective}, $N_{\mathrm{c}}$ such inequalities
form the constrained problem
\begin{equation}
\max_{\param\in\feasset}\;U(\param)\quad
\text{s.t.}\quad g_n(\param)\!\le\!R_n,\ n\!=\!1,\dots,N_{\mathrm{c}},
\label{eq:cmi_constraint}
\end{equation}
where each $g_n$ is a conditional MI on $\mathcal{G}$, or more
generally a linear combination of conditional MIs forming an
aggregate rate budget, and $R_n$ a prescribed bound. Standard
single-MI instances are leakage budgets in physical-layer security
($\MI(V_X;V_Z)\!\le\!\varepsilon$), per-link fronthaul caps
in cell-free networks
($\MI(V_{X_k};V_{\hat{X}_k})\!\le\!C_{\mathrm{FH}}$ for
each $k$), and information-bottleneck constraints
($\MI(V_U;V_X)\!\le\!R_{\mathrm{IB}}$); aggregate forms
such as a total fronthaul cap
$\sum_k\MI(V_{X_k};V_{\hat{X}_k})\!\le\!C_{\mathrm{total}}$
fit the same template.

For any multipliers $\bm{\lambda}\!=\!(\lambda_1,\dots,\lambda_{N_{\mathrm{c}}})$
with $\lambda_n\!\ge\!0$, the Lagrangian of~\eqref{eq:cmi_constraint}
reduces, up to a constant $\sum_n\lambda_n R_n$ that is independent
of $\param$ and immaterial for the inner $\param$-optimization, to
the $\param$-dependent part
\begin{equation}
U_{\bm{\lambda}}(\param)
\!=\!U(\param)
\!-\!\sum_{n=1}^{N_{\mathrm{c}}}\!\lambda_n\,g_n(\param),
\label{eq:cmi_lagrangian}
\end{equation}
which is itself a sign-indefinite conditional-MI
objective~\eqref{eq:objective} (and linear,~\eqref{eq:utility}, when
$U$ is): every $g_n$ is evaluated from the K-blocks
of~\cref{prop:cmi_schur} on the same shared graph, and one
reverse-mode AD sweep returns the Lagrangian gradient
$\nabla_{\param^{*}}U_{\bm{\lambda}}$ needed for the primal update.
Conditional-MI constraints therefore introduce no new computational
primitive.

\subsection{Shared Differentiable Computation Graph}
\label{sec:cmi_graph}

The objectives~\eqref{eq:objective}, the constraint MIs
of~\eqref{eq:cmi_constraint}, and their non-linear composites
(Appendix~\ref{app:rate_region_examples}) are all built from one
primitive: a conditional MI $\MI(V_A;V_B\,|\,V_C)$ on the linear
Gaussian DAG $\mathcal{G}$, evaluated from the K-blocks
of~\cref{prop:cmi_schur}. The whole design problem thus sits on a
\emph{single shared computation graph} rooted at $\param$, on which
each gradient evaluation is the pipeline:
\begin{enumerate}[label=(\arabic*),leftmargin=*,nosep,topsep=2pt]
\item one K-recursion forward pass~\eqref{eq:krec} computing the
K-blocks from $\param$;
\item block extractions, Schur complements, and log-determinants for
every CMI requested by the scalar objective $U$
(or Lagrangian $U_{\bm{\lambda}}$);
\item formation of the scalar $U$ or $U_{\bm{\lambda}}$;
\item one reverse-mode AD backward sweep from that scalar, returning
$\nabla_{\param^{*}}U$ (or $\nabla_{\param^{*}}U_{\bm{\lambda}}$)
with respect to all controllable factors (\cref{sec:wirtinger}).
\end{enumerate}
Steps~(1) and~(4) each run once per evaluation, independent of the
number $N$ of CMIs requested or of the choice of $\Phi$; the
per-CMI cost lies in step~(2). The resulting gradient feeds the
projected gradient update of~\cref{sec:opt}.

\subsection{Optimization of Conditional-MI Objectives}
\label{sec:opt}

The framework targets optimization problems of the form
\begin{equation}
\max_{\param\in\feasset}\;U(\param)
\quad \text{s.t.}\quad
g_n(\param)\!\le\!R_n,\ n=1,\dots,N_{\mathrm{c}},
\label{eq:opt_problem}
\end{equation}
combining (i) a conditional-MI objective $U$ of~\eqref{eq:objective}
in the design parameter $\param$ of~\eqref{eq:design_param}, (ii) the
physical feasible set $\feasset$ of~\cref{sec:krec_review}, which
encodes Frobenius or total-power budgets, diagonal or unit-modulus
parameterizations, and similar structural constraints on the
controllable factors, and (iii) $N_{\mathrm{c}}\!\ge\!0$ conditional-MI
constraints~\eqref{eq:cmi_constraint} when present. The constraint
term in~\eqref{eq:opt_problem} is absent for $N_{\mathrm{c}}\!=\!0$,
and each $g_n$ on the constraint side may be a single CMI
(e.g., a leakage budget) or a linear combination of CMIs (e.g., an
aggregate fronthaul cap); see \cref{sec:cmi_constraints}. The two
constraint types are handled differently by the two steps described
below: $\feasset$ is enforced by Euclidean projection in every
iterate (\cref{sec:opt_pgd}), while the conditional-MI constraints
are absorbed into the objective via Lagrangian relaxation
(\cref{sec:opt_lag}).

\subsubsection{Projected Gradient Method}
\label{sec:opt_pgd}
We first treat the case $N_{\mathrm{c}}\!=\!0$, where
\eqref{eq:opt_problem} reduces to maximization of a conditional-MI
objective over $\feasset$. The objective $U$ of~\eqref{eq:objective}
is \emph{maximized} when it scores a communication benefit (e.g.,
a weighted sum-rate) and \emph{minimized} when it scores a cost
(e.g., a rate-region outage surrogate). Either direction is handled
by the projected gradient method: with step size $\alpha_t\!>\!0$,
\begin{equation}
\param^{(t+1)}\!=\!\Pi_{\feasset}\!\bigl(\param^{(t)}\!\pm\!\alpha_t\,
\nabla_{\param^{*}}U(\param^{(t)})\bigr),
\label{eq:pgd}
\end{equation}
the $+$ sign giving gradient ascent for maximization and the $-$ sign
descent for minimization, where $\Pi_{\feasset}$ is the Euclidean
projection onto $\feasset$, in closed form for the constraint
families of~\cref{sec:krec_review}. Each iteration is one K-recursion forward pass,
one reverse-mode AD backward pass, and one closed-form projection;
the gradient is exact up to the Wirtinger convention
of~\cref{sec:wirtinger}. Since $U$ is in general nonconvex in
$\param$ and $\feasset$ may itself be nonconvex (e.g., unit-modulus
or hybrid analog/digital structures), the method should be
understood as seeking a stationary point: under standard smoothness
and step-size assumptions, any limit point of the iterates
satisfying the projected first-order condition is stationary.

\subsubsection{Constrained Objectives via Lagrangian Relaxation}
\label{sec:opt_lag}
When $N_{\mathrm{c}}\!\ge\!1$ conditional-MI constraints are present
in~\eqref{eq:opt_problem}, the
Lagrangian~\eqref{eq:cmi_lagrangian} of~\cref{sec:cmi_constraints}
is itself a sign-indefinite conditional-MI objective (a
sign-indefinite linear instance~\eqref{eq:utility} when $U$ is
linear), so for any fixed multiplier vector
$\bm{\lambda}\!\ge\!\bm{0}$ the projected-gradient
update~\eqref{eq:pgd} applies to $U_{\bm{\lambda}}$ unchanged (the structural feasible set $\feasset$
is enforced as before). To
trace the empirical Pareto curve between $U$ and the constraint
quantities $g_n$, $\bm{\lambda}$ is swept on a grid; at each
$\bm{\lambda}$ several initial points (a warm start from the
previous $\bm{\lambda}$ together with random complex-Gaussian
restarts) are optimized in parallel and the candidate of highest
$U_{\bm{\lambda}}$ is retained, mitigating the local optima of the
non-convex Lagrangian. Under the same standard smoothness and
step-size assumptions, any limit point is a local KKT point
of~\eqref{eq:cmi_constraint} rather than a global optimum, so
sweeping $\bm{\lambda}$ traces \emph{locally supported}
Pareto-stationary points; for any active constraint
($\lambda_n\!>\!0$), complementary slackness ties the multiplier to
the constraint level via $R_n\!=\!g_n(\param^{*}_{\bm{\lambda}})$.
The wiretap experiment of~\cref{sec:exp_secrecy} implements this
procedure for a single constraint.

\begin{remark}[Point-to-point special case]
\label{rem:p2p}
With a single rate-bearing user, $K\!=\!1$, a rate function reduces to
one unconditional mutual information $\MI(V_X;V_Y)$ and the
framework degenerates to the exact information-gradient optimization
of~\cite{wadayama2026arxiv}; the conditional-MI
objective~\eqref{eq:objective} is its multi-terminal generalization.
\end{remark}

\section{Numerical Results}
\label{sec:experiments}

We demonstrate the framework with three experiments, each on a fixed
linear Gaussian DAG, the deterministic setting of~\cref{sec:objectives}.
Each optimizes a conditional-MI objective end to end by the projected
gradient method, with no problem-specific gradient derivation: every
iteration is one K-recursion forward pass and one reverse-mode AD
sweep. \Cref{tab:num_params} collects the numerical parameters used
in each experiment. All runs use PyTorch \texttt{complex128} on CPU,
a constant projected-gradient step size, and Frobenius-ball
projection by uniform rescaling; no Cholesky jitter is activated,
since the support covariances stay well-conditioned in double
precision. The random seed for~\cref{sec:exp_random_mac} fixes the
single network instance used throughout---its layered inter-layer
edge set and the per-edge i.i.d.\ $\CN(0,1)$ channels, instantiated
once and held fixed. The $\lambda$-sweep of~\cref{sec:exp_secrecy}
uses $\lambda\!\in\!\{0,0.1,0.25,0.5,0.75,1,1.25,1.5,2,3,5\}$ with
$10$ multi-start restarts per $\lambda$ (one warm start from the
previous $\lambda$ and nine random complex-Gaussian initial points
rescaled to $\|\bm{F}\|_F^2\!=\!P$). The numerical implementation
underlying these experiments (the K-recursion forward pass, the
sub-block Schur-complement evaluator of~\cref{prop:cmi_schur}, the
projected-gradient loop, and the PyTorch autograd interface) is
available as the open-source library \texttt{cmi-dag} at
\url{https://github.com/wadayama/cmi-dag}; its symbolic companion
\texttt{symbolic-dag}, which derives the same conditional-MI
quantities and their Wirtinger gradients in closed form and
cross-validates them against the numerical evaluator, is described
in Appendix~\ref{app:symbolic_dag}.

\begin{table}[!t]
\centering
\caption{Numerical parameters of the three experiments.}
\label{tab:num_params}
\renewcommand{\arraystretch}{1.1}
\begin{tabular}{@{}lccc@{}}
\toprule
                       & \cref{sec:exp_rate_region}
                       & \cref{sec:exp_secrecy}
                       & \cref{sec:exp_random_mac} \\
\midrule
Antenna dim $d$        & 4     & 4     & 4     \\
Power budget $P$       & 8     & 8     & 36    \\
Iterations $T$         & 120   & 200   & 800   \\
Step size $\alpha_t$   & 0.01  & 0.04  & 0.003 \\
Random seed            & 7     & 7     & 7     \\
Jitter $\epsilon$      & 0     & 0     & 0     \\
$\lambda$-sweep        & ---   & yes   & ---   \\
\bottomrule
\end{tabular}
\end{table}

\subsection{Rate-Region Maximization}
\label{sec:exp_rate_region}
The first experiment enlarges the achievable rate region of the
two-user MIMO MAC of~\cref{ex:mac}.

\emph{Setup.}
The channel is one fixed realization of a two-user MAC with
$d\!=\!4$ antennas per transmitter and at the receiver, i.i.d.\
$\CN(0,1)$ entries, and receiver noise $\sigma^2\!=\!1$. The design
parameter is the precoder pair $(\bm{F}_1,\bm{F}_2)$, constrained to
the shared total-power ball
$\|\bm{F}_1\|_F^2\!+\!\|\bm{F}_2\|_F^2\!\le\!P$ with $P\!=\!8$.
Writing $\MI_1,\MI_2,\MI_{12}$ for the three pentagon facets
of~\cref{ex:mac}, the objective is the linear conditional-MI
objective~\eqref{eq:utility}
\begin{equation*}
U(\param)=\MI_1+\MI_2+\MI_{12},
\end{equation*}
whose positive-weight maximization encourages simultaneous
enlargement of the three MAC facets of the rate
region~\eqref{eq:rate_region}. We maximize $U$ by the
projected gradient method~\eqref{eq:pgd} from the uniform precoders
$\bm{F}_k\!=\!\bm{I}_d$ (no precoding, equal power, already on the
budget boundary), with a constant step size; projection onto the
total-power ball is closed-form.

\emph{Results.}
\Cref{fig:rate_region} reports the outcome. Panel~(a) draws the MAC
rate-region pentagon at five iterations: in this realization all
three facets increase monotonically along the trajectory, so the
rate region expands outward from the unprecoded region at
iteration~$0$ to the optimized one at iteration~$120$. Panel~(b) shows $U$ rising
smoothly from $20.05$ to $22.60$ nats and converging within about
$120$ iterations. The facets grow from
$(\MI_1,\MI_2,\MI_{12})\!=\!(6.30,4.93,8.82)$ to $(7.33,5.16,10.11)$
nats, and the rate-region pentagon area\footnote{For a two-user MAC
pentagon $\{(R_1,R_2)\!\in\!\R_+^2:R_1\!\le\!\MI_1,\,R_2\!\le\!\MI_2,\,
R_1\!+\!R_2\!\le\!\MI_{12}\}$ in the regime
$\max(\MI_1,\MI_2)\!\le\!\MI_{12}\!\le\!\MI_1\!+\!\MI_2$ (satisfied
along the entire trajectory here), the area is
$\MI_1\MI_2-\tfrac{1}{2}(\MI_1\!+\!\MI_2\!-\!\MI_{12})^2$.}
increases by a factor~$1.24$.
As an external benchmark, waterfilling on the stacked channel
$[\bm{H}_1\,\bm{H}_2]$~\cite{telatar1999} under the same total power
$P$ yields the cooperative MAC sum-capacity
$C_{\mathrm{coop}}\!=\!11.32$~nats, an upper bound on $\MI_{12}$ since
cooperation removes the MAC's independent-encoder restriction; the
optimized facet $\MI_{12}\!=\!10.11$ attains $89\%$ of this bound,
leaving a $1.21$-nat residual gap.

\emph{Discussion.}
The total power stays at $P$ throughout (the uniform initialization
already lies on the budget boundary), so the entire gain is due to
precoder \emph{shaping}: at an unchanged power budget, the framework
finds the precoders that best match the two users to the fixed
channel. No problem-specific gradient was derived; the objective $U$
was simply specified, and its Wirtinger gradient
$\nabla_{\param^{*}}U$ was returned by one reverse-mode AD sweep
through the shared K-recursion graph at each iteration.

\begin{figure}[!tb]
\centering
\includegraphics[width=\linewidth]{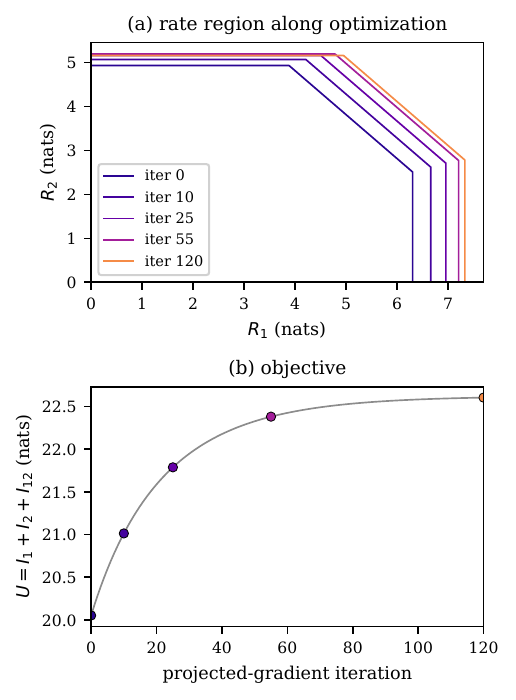}
\caption{Rate-region maximization on a fixed two-user MIMO MAC
($d\!=\!4$, shared power budget $P\!=\!8$). The precoders are
optimized to maximize the linear conditional-MI objective
$U\!=\!\MI_1\!+\!\MI_2\!+\!\MI_{12}$ by the projected gradient method.
(a)~the MAC rate-region pentagon at selected iterations, expanding
outward as the precoders are optimized; (b)~the objective $U$ versus
projected-gradient iteration.}
\label{fig:rate_region}
\end{figure}

\subsection{Secure Precoding}
\label{sec:exp_secrecy}
The second experiment optimizes a \emph{sign-indefinite} conditional-MI
objective (a secrecy rate) on a MIMO wiretap channel.

\emph{Setup.}
A transmitter $X$ sends to a legitimate receiver $Y$ while an
eavesdropper $Z$ observes the same transmission, a linear Gaussian
DAG with edges $X\!\to\!Y$ and $X\!\to\!Z$. The channels are one fixed
realization with $d\!=\!4$ antennas at every terminal (the
eavesdropper is \emph{not} assumed weaker), i.i.d.\ $\CN(0,1)$
entries, receiver noise $\sigma^2\!=\!1$. The design parameter is the
precoder $\bm{F}$, constrained to the Frobenius ball
$\|\bm{F}\|_F^2\!\le\!P$ with $P\!=\!8$. The objective is the secrecy
rate
\begin{equation*}
U(\param)=\MI(V_X;V_Y)-\MI(V_X;V_Z),
\end{equation*}
a conditional-MI objective~\eqref{eq:objective} whose linear $\Phi$ is
\emph{sign-indefinite}: the negative weight on the eavesdropper's
information is a leakage penalty. We maximize $U$ by the projected
gradient method~\eqref{eq:pgd} from the uniform precoder
$\bm{F}\!=\!\bm{I}_d$.

\emph{Results.}
\Cref{fig:secrecy} reports the outcome. Panel~(a) tracks the two
mutual informations along the iterations: $\MI(V_X;V_Y)$
moves from $8.23$ to $5.74$ nats, while $\MI(V_X;V_Z)$
drops more sharply from $6.72$ to $2.06$ nats (a $69\%$ reduction).
Panel~(b) shows the resulting secrecy rate climbing from $1.51$ to
$3.67$ nats, a factor of $2.43$. Panel~(c) extends the
optimization to a \emph{Lagrangian sweep} over $\lambda\!\ge\!0$ in
$U_\lambda\!=\!\MI(V_X;V_Y)\!-\!\lambda\,\MI(V_X;V_Z)$,
with $10$ multi-start restarts per $\lambda$ (one warm-start
from the previous $\lambda$ and nine random complex-Gaussian inits),
keeping the candidate of highest $U_\lambda$ as the local stationary
point at that $\lambda$. The resulting curve is an \emph{empirical
leakage-rate Pareto curve} produced by the proposed local
optimization procedure; consistently with~\cref{sec:opt_lag}, it
should be read as a set of locally supported Pareto-stationary
points along the leakage-rate trade-off rather than as the global
frontier $\max\{\MI(V_X;V_Y):\MI(V_X;V_Z)\!\le\!R\}$. The classical
secrecy rate ($\lambda\!=\!1$, the star marker) is a single point of
this curve.

\emph{Discussion.}
Both mutual informations decrease as the precoder restructures, but
$\MI(V_X;V_Z)$ drops far more steeply, so the secrecy rate
$\MI(V_X;V_Y)\!-\!\MI(V_X;V_Z)$ more than doubles
even though $\MI(V_X;V_Y)$ is partially traded away: the
precoder reallocates transmit power toward modes with a favorable
legitimate/eavesdropper gain tradeoff. Panel~(c) puts this single
operating
point in the broader context of a \emph{constraint-side} use of
conditional MI: for any prescribed leakage budget $R$, the
locally optimized legitimate rate produced by the procedure is read
off the empirical Pareto curve at $\MI(V_X;V_Z)\!=\!R$ via the
Lagrangian relaxation of~\cref{sec:cmi_constraints}, solved by the
projected-gradient update of~\cref{sec:opt_lag}. As in~\cref{sec:exp_rate_region}, no gradient
was derived by hand: only the scalar $\Phi$ changed, here to the
$\lambda$-parameterized sign-indefinite linear form, with the same
single-line K-recursion graph driving every $\lambda$.
\Cref{sec:exp_rate_region,sec:exp_secrecy} together exercise the
conditional-MI objective class of~\cref{sec:objectives} across
positive-weight, sign-indefinite, and constraint-side objectives.

\begin{figure}[!tb]
\centering
\includegraphics[width=\linewidth]{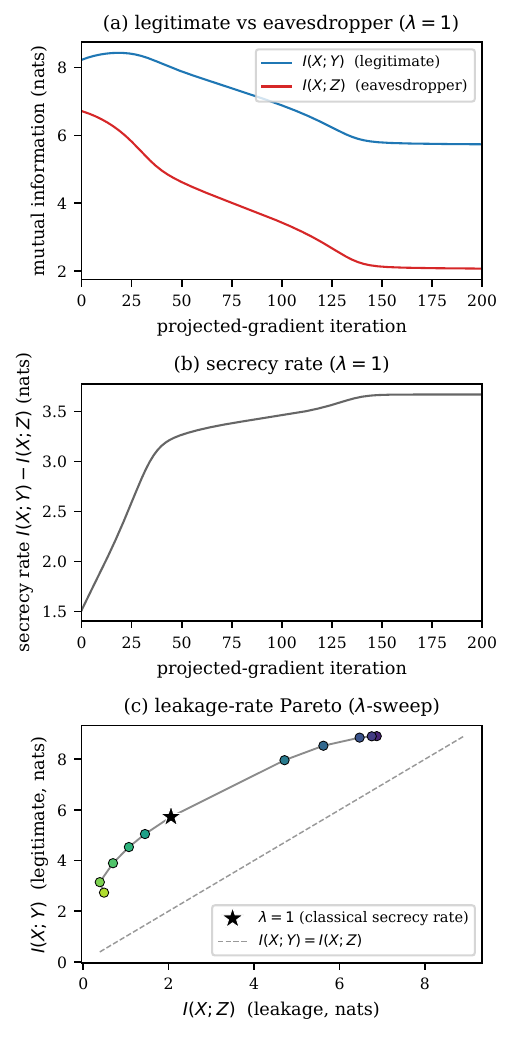}
\caption{Secure precoding on a fixed MIMO wiretap channel ($d\!=\!4$
antennas at every terminal, power budget $P\!=\!8$).
Panels~(a),~(b): the precoder maximizes the secrecy rate
$\MI(V_X;V_Y)\!-\!\MI(V_X;V_Z)$ ($\lambda\!=\!1$)
by the projected gradient method; both informations decrease as the
precoder restructures, but the eavesdropper's information drops more
steeply, so the secrecy rate more than doubles.
Panel~(c): empirical leakage-rate Pareto curve traced by a Lagrangian
sweep $U_\lambda\!=\!\MI(V_X;V_Y)\!-\!\lambda\,\MI(V_X;V_Z)$ with
$10$ multi-start restarts per $\lambda$ (locally supported
Pareto-stationary points; \cref{sec:opt_lag}); the star marks the
$\lambda\!=\!1$ point of the sweep, the classical secrecy-rate
optimization shown in panels~(a),~(b).}
\label{fig:secrecy}
\end{figure}

\subsection{Application to a Multi-Hop MAC Network}
\label{sec:exp_random_mac}
The third experiment takes the rate-region objective
of~\cref{sec:exp_rate_region} from the canonical single-hop MAC to a
larger \emph{multi-hop} network, illustrating that the
same procedure applies without modification to a nontrivial
topology.

\emph{Setup.}
The DAG is a given layered Gaussian
multiple-access network: two sources $s_1,s_2$, three relay layers of
three nodes each, and a single sink $t$ ($M=12$ nodes and $19$ edges,
with $d=4$ per node), with a fixed set of inter-layer edges
(\cref{fig:random_mac}(a)). The topology is given; only the per-edge
channels are random, each a fixed i.i.d.\ $\CN(0,1)$ realization drawn
once from a fixed seed. The two sources emit isotropic signals, and each of the nine
relay nodes carries a controllable processing matrix $\bm{F}_i$. All
relay matrices share one global total-power budget
$\sum_i\|\bm{F}_i\|_F^2\!\le\!P$ with $P\!=\!36$. The objective is the
MAC rate-region facet sum
\begin{equation*}
U(\param)=
\underbrace{\MI(V_{s_1};V_t\mid V_{s_2})}_{\MI_1}
+\underbrace{\MI(V_{s_2};V_t\mid V_{s_1})}_{\MI_2}
+\underbrace{\MI(V_{s_1},V_{s_2};V_t)}_{\MI_{12}},
\end{equation*}
the same linear conditional-MI objective as~\cref{sec:exp_rate_region}
(the two single-user facets $\MI_1,\MI_2$ are conditional mutual
informations, $\MI_{12}$ the joint sum-rate facet), now on a
multi-hop network with no simple topology-specific closed-form
optimizer. We maximize $U$ by
the projected gradient method~\eqref{eq:pgd} from the uniform
allocation (identity processing at every relay, equal power;
$P\!=\!(M\!-\!3)d$ places this on the budget boundary).

\emph{Results.}
\Cref{fig:random_mac} reports the outcome. Panel~(a) draws the
multi-hop MAC network, each relay node shaded by its optimized power
$\|\bm{F}_i\|_F^2$, a non-uniform, network-wide allocation. Panel~(b)
shows the objective and the three facets versus iteration: $U$
rises from $10.44$ to $18.62$ nats, a factor of $1.78$, and
converges; in this realization all three facets also increase,
$(\MI_1,\MI_2,\MI_{12})$ from $(3.64,1.80,5.00)$ to
$(6.59,2.98,9.05)$ nats, so the achievable rate region expands
outward.

\emph{Discussion.}
A multi-hop network of this size has no simple topology-specific
closed-form optimizer or hand-derived facet gradient. Specifying the facet sum $U$ is enough:
one reverse-mode AD sweep through the $M$-node K-recursion graph
returns the Wirtinger gradient with respect to all nine relay matrices
jointly, and projection onto the shared budget carries the
network-wide power allocation. Whereas~\cref{sec:exp_rate_region}
optimizes this rate-region objective on the canonical single-hop MAC,
here the \emph{same} objective applies without modification to a
multi-hop multi-terminal network of $M\!=\!12$ nodes and
$19$ edges, illustrating that the framework reaches linear Gaussian
DAGs of nontrivial size and topology with no new derivation.

\begin{figure}[!tb]
\centering
\includegraphics[width=\linewidth]{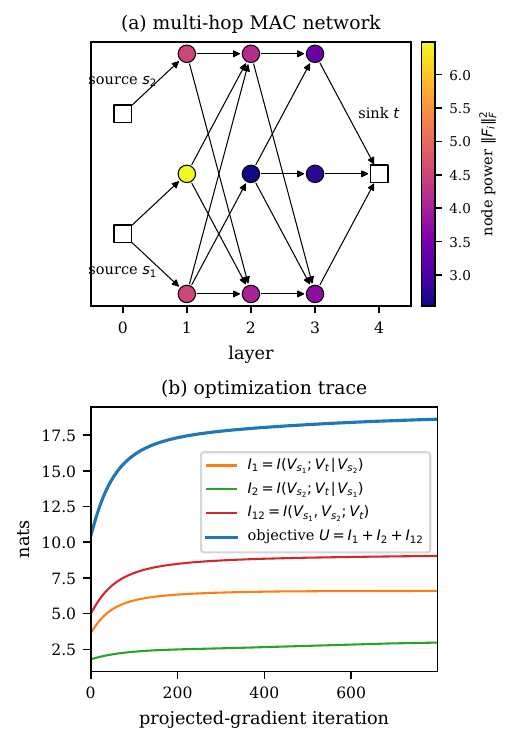}
\caption{Rate-region maximization on a fixed multi-hop Gaussian
MAC network ($M\!=\!12$ nodes, $19$ edges, $d\!=\!4$ per node, shared
power budget $P\!=\!36$). Nine relay processing matrices are optimized
jointly to maximize the facet sum $U\!=\!\MI_1\!+\!\MI_2\!+\!\MI_{12}$
by the projected gradient method. (a)~the MAC network (two
sources at the left, three relay layers, one sink at the right;
edges are directed left to right, reflecting the DAG topology),
each relay shaded by its optimized power $\|\bm{F}_i\|_F^2$; (b)~the objective $U$ and the three facets
$\MI_1,\MI_2,\MI_{12}$ versus projected-gradient iteration.}
\label{fig:random_mac}
\end{figure}

\section{Conclusion}
\label{sec:conclusion}

We presented a differentiable framework for
conditional-mutual-information design over linear Gaussian DAGs. Its
analytic core (\cref{prop:cmi_schur}) expresses any
conditional mutual information on the DAG as a log-determinant
difference of sub-block Schur complements of the support covariance,
assembled by a single K-recursion forward pass. On the
positive-definite Gaussian-DAG regime considered here (with the
diagonal-jitter regularization of~\cref{rem:cholesky} in numerical
implementations), every differentiable function of conditional
mutual informations (weighted rate objectives, secrecy or leakage
objectives, rate-region functions, and smooth non-linear surrogates
such as outage approximations) is an end-to-end differentiable
function of the design parameters and amenable to projected-gradient
optimization. The main practical consequence is that the user
specifies the DAG and the CMI-based scalar objective, while the
gradient computation is delegated to the shared K-recursion /
Wirtinger AD graph, with no problem-specific gradient derivation.

We demonstrated the framework on three experiments that together
cover positive-weight rate-region objectives, sign-indefinite
secrecy objectives with a Lagrangian leakage sweep, and shared
relay-matrix optimization on a nontrivial multi-hop DAG, all by the
projected gradient method; the wiretap Lagrangian sweep further
illustrates the constraint-side use of conditional MI.

Two complementary directions follow naturally.
\emph{(i) Within the present wireless/network optimization scope:}
the composite sigmoid surrogate of rate-region outage, its
stochastic optimization under statistical and instantaneous CSI,
and applications to broader multi-terminal channels beyond the MAC
and wiretap settings explored here.
\emph{(ii) Beyond CMI, as an information primitive:} the same
K-recursion graph supports differentiable conditional differential
entropy as a one-Schur, one-log-det specialization
(Appendix~\ref{app:entropy}), opening Bayesian experimental design,
distributed source coding, and Gaussian process innovation rates,
which are treated in a follow-up paper.

More broadly, the rich rate-region results of multi-terminal
information theory have not always translated naturally into the
gradient-based workflows of modern wireless-system design. By
rendering any combination of conditional mutual informations
end-to-end differentiable on a shared K-recursion graph, we hope
the present framework offers a practical bridge between the
rate-function calculus of multi-terminal information theory and
physical-layer design as differentiable programming.

\appendices
\section{Multi-Terminal Rate-Region Examples}
\label{app:rate_region_examples}

This appendix collects concrete instances of the abstract
forms~\eqref{eq:objective} and~\eqref{eq:utility} drawn from
multi-terminal information theory: the rate-function facets of
standard rate regions (a linear instance of~\eqref{eq:utility}), the
composite sigmoid surrogate of rate-region outage (a non-linear
instance of~\eqref{eq:objective}), and a worked two-user MAC example.

For $K$ user rates $\bm{R}\!=\!(R_1,\dots,R_K)$, with the coding
architecture, auxiliary-variable structure, and any time-sharing
choice fixed, the achievable region at design $\param$ is an
intersection of finitely many $\log\det$
inequalities~\cite{cover2006,elgamal2011},
\begin{equation}
\rateset(\param)\!=\!
\Big\{\bm{R}\!\in\!\R_+^K:\!\sum_{k\in T}\!R_k\!\le\!f_T(\param),\,
\forall T\!\in\!\indset\Big\},
\label{eq:rate_region}
\end{equation}
indexed by a family
$\indset\!\subseteq\!2^{[K]}\!\setminus\!\{\emptyset\}$, each rate
function $f_T$ being itself a linear conditional-MI
objective~\eqref{eq:utility},
\begin{equation}
f_T(\param)\!=\!\sum_{n=1}^{N_T}\alpha_{T,n}\,
\MI(V_{A_{T,n}};V_{B_{T,n}}\,|\,V_{C_{T,n}}).
\label{eq:fT}
\end{equation}
The two-user multiple-access pentagon is the case $N_T\!=\!1$,
$f_T\!=\!\MI(V_{X_T};V_Y\,|\,V_{X_{T^c}})$~\cite{tse1998multiaccess};
the Gaussian broadcast region with dirty-paper coding~\cite{caire2003},
the Han--Kobayashi inner bound for the interference
channel~\cite{hankobayashi1981}, and decode-/compress-and-forward
relay inner bounds~\cite{cover1979relay} likewise admit the
form~\eqref{eq:rate_region}--\eqref{eq:fT} once the coding order,
auxiliary-variable structure, time-sharing choice, and Gaussian
linear parametrization are fixed, the latter two with $N_T\!\geq\!2$
and coefficients $\alpha_{T,n}$ of either sign.

A notable \emph{non-linear} instance of~\eqref{eq:objective} is built
directly from these rate functions. At an operating point $\bm{R}$,
the achievability indicator $\bm{1}\{\bm{R}\!\in\!\rateset(\param)\}$
factors as
$\prod_{T\in\indset}\bm{1}\{\sum_{k\in T}R_k\!\le\!f_T(\param)\}$;
smoothing each factor by a temperature-$\tau$ sigmoid $\sigtau$ yields
the \emph{composite sigmoid surrogate} of the rate-region outage
indicator,
\begin{equation}
\rhotau(\param)\!=\!1\!-\!\!\prod_{T\in\indset}\!
\sigtau\!\Bigl(f_T(\param)\!-\!\!\sum_{k\in T}\!R_k\Bigr),
\label{eq:composite}
\end{equation}
whose $\Phi$ is the non-linear sigmoid-product composition and whose
auxiliary constants are the target rates $\bm{R}$ and the temperature
$\tau$. In a fading environment, with the channel realization $\bm{H}$
made an explicit random argument, the expectation
$\E_{\bm{H}}[\rhotau(\param,\bm{H})]$ is a differentiable proxy for
the rate-region outage probability
$\Pr_{\bm{H}}[\bm{R}\!\notin\!\rateset(\param,\bm{H})]$.

\begin{example}[Two-user multiple-access channel]
\label{ex:mac}
Two rate-bearing users with input roots $X_1,X_2$ feed a single
receiver node $Y$. Three conditional mutual informations on this DAG
generate the standard design objectives, namely the interference-free user
rates $\MI_1\!\triangleq\!\MI(V_{X_1};V_Y\,|\,V_{X_2})$
and $\MI_2\!\triangleq\!\MI(V_{X_2};V_Y\,|\,V_{X_1})$,
and the sum information
$\MI_{12}\!\triangleq\!\MI(V_{X_1},V_{X_2};V_Y)$:
\emph{(i)}~sum throughput is $U\!=\!\MI_{12}$;
\emph{(ii)}~a weighted sum-rate $U\!=\!\alpha_1\MI_1\!+\!\alpha_2\MI_2$
is a linear conditional-MI objective~\eqref{eq:utility};
\emph{(iii)}~proportional fairness $U\!=\!\log\MI_1\!+\!\log\MI_2$ is a
non-linear one~\eqref{eq:objective};
\emph{(iv)}~at an operating point $(R_1,R_2)$ the composite
surrogate~\eqref{eq:composite} of the multiple-access pentagon is
$\rhotau\!=\!1\!-\!\sigtau(\MI_1\!-\!R_1)\,\sigtau(\MI_2\!-\!R_2)\,
\sigtau(\MI_{12}\!-\!R_1\!-\!R_2)$.
All four are read off the one K-recursion graph, differing only in
the node subsets and in $\Phi$, and optimized by the single
update~\eqref{eq:pgd}.
\end{example}

\section{Conditional Differential Entropy as a Pipeline Specialization}
\label{app:entropy}

The conditional-MI evaluator of \cref{prop:cmi_schur} specializes
naturally to the conditional differential entropy. For nonempty
$A\!\subseteq\!\mathcal{V}$ and a (possibly empty) disjoint
$C\!\subseteq\!\mathcal{V}$, set $S\!=\!A\!\cup\!C$ and assume
$\bm{\Sigma}_{S,S}\!\succ\!\bm{0}$ as in \cref{prop:cmi_schur}. The
conditional covariance $\bm{\Sigma}_{A|C}$ of $V_A$ given $V_C$,
defined by the Schur complement of~\eqref{eq:schur} (reducing to
$\bm{\Sigma}_{A,A}$ for $C\!=\!\emptyset$), yields the conditional
differential entropy
\begin{equation}
h(V_A\,|\,V_C)
\!=\!\log\det(\pi e\,\bm{\Sigma}_{A|C})
\!=\!\log\det\bm{\Sigma}_{A|C}\!+\!d_A\log(\pi e),
\label{eq:cond_entropy}
\end{equation}
where $d_A\!=\!\sum_{a\in A}d_a$ is the total dimension of $V_A$.
Since $d_A$ is fixed by the support of $A$ and $\pi e$ is a global
constant, the additive term $d_A\log(\pi e)$ is independent of the
design parameter $\param$; the entropy gradient consequently reduces
to a pure log-determinant gradient,
\begin{equation}
\nabla_{\param^{*}}h(V_A\,|\,V_C)
\!=\!\nabla_{\param^{*}}\log\det\bm{\Sigma}_{A|C},
\label{eq:entropy_gradient}
\end{equation}
so the constant may be dropped throughout any gradient-based design
objective involving $h$.

The construction is the first half of~\eqref{eq:cmi_logdet}: one
block extraction, one Schur complement, and one log-determinant.
The conditional MI of~\cref{prop:cmi_schur} is recovered as the
entropy chain rule
\begin{equation}
\MI(V_A;V_B\,|\,V_C)\!=\!h(V_A\,|\,V_C)\!-\!h(V_A\,|\,V_{BC}),
\label{eq:entropy_chain}
\end{equation}
in which the constant $d_A\log(\pi e)$ cancels exactly.

The conditional-entropy primitive is of independent interest in
problems where the MI representation does not collapse to fewer
quantities and the entropy itself enters the formulation: Bayesian
experimental design and sensor placement (posterior covariance-volume
minimization via $\log\det\bm{\Sigma}_{A|C}$); Gaussian Slepian--Wolf
type distributed source coding (rate regions in terms of conditional
entropies); innovation rates of Gaussian processes on a DAG (Kalman
residual entropy); privacy and information-leakage formulations
where input distributions are co-designed (so that $h(V_X)$ itself
is a variable); and maximum-entropy design under structural
constraints. These applications fit the same K-recursion / Wirtinger
AD machinery through the specialization~\eqref{eq:cond_entropy} but
are outside the scope of the present paper and are left to a
follow-up work.

\section{Symbolic Companion: Closed-Form CMI and Wirtinger Gradients}
\label{app:symbolic_dag}

The K-recursion and the conditional-MI construction
of this paper are purely linear-algebraic and therefore admit an
exact \emph{symbolic} treatment. We have implemented this in a small
open-source companion library, \texttt{symbolic-dag}, which shares
the model and conventions of~\cref{sec:krec_review} but carries the
edge gains and noise covariances as \emph{opaque matrix symbols}.
Its role is analysis rather than optimization: while the numerical
pipeline of~\cref{sec:cmi,sec:experiments} searches the design
space, the symbolic layer supplies closed-form expressions,
optimality conditions, and exact conditional-independence
identities. Kept symbolic, the conditional mutual
information~\eqref{eq:cmi_logdet} admits the equivalent symmetric
rewriting
\begin{equation}
\begin{aligned}
\MI(V_A;V_B\,|\,V_C)
\!=\!&\log\det\bm{\Sigma}_{A|C}\!+\!\log\det\bm{\Sigma}_{B|C}\\
&-\!\log\det\bm{\Sigma}_{AB|C},
\end{aligned}
\label{eq:cmi_symmetric}
\end{equation}
which is returned as a closed form for an arbitrary
disjoint triple $(V_A,V_B,V_C)$ whose size is independent of the
node dimension (the scalar case is its $1\!\times\!1$
specialization), together with its Wirtinger gradient
$\partial\MI/\partial\bm{F}^{\herm}$ with respect to any gain or
precoder $\bm{F}$.

For the MIMO precoder model
$V_Y\!=\!\bm{H}\bm{F}V_{X_0}\!+\!V_{X_1}\!+\!V_N$, for instance,
the library returns the closed-form gradient
\begin{equation}
\frac{\partial}{\partial\bm{F}^{\herm}}
\MI(V_{X_0};V_Y\,|\,V_{X_1})
\!=\!\bm{H}^{\herm}\!\bigl(\bm{R}\!+\!\bm{H}\bm{F}\bm{\Sigma}_0\bm{F}^{\herm}\bm{H}^{\herm}\bigr)^{-1}\!\bm{H}\bm{F}\bm{\Sigma}_0,
\label{eq:symbolic_gradient}
\end{equation}
for any precoder $\bm{F}$, with no per-topology hand calculation.
Conditional
independence is similarly obtained as an exact matrix identity: the
cross conditional covariance $\bm{\Sigma}_{AB|C}$ reduces to the
zero matrix for all dimensions at once, supporting matrix-level
$d$-separation proofs.

The central technical point is that a
general-purpose computer-algebra system does not, out of the box,
operate at this matrix layer: in the SymPy base layer used here,
covariance symbols carry no Hermitian structure, the
block-determinant identities that collapse a conditional-MI
log-determinant are not applied, and the native matrix
differentiation of $\log\det$ of an opaque matrix symbol returns
the zero matrix. \texttt{symbolic-dag} therefore supplies two
small, rule-based engines. The first is a \emph{simplification
engine}: a strategic (phased) rewriter over a rule set of block
identities (Schur complement, Sylvester's determinant identity, the
matrix-determinant lemma, and Woodbury) together with structural
assumptions (Hermitian symmetry
$\bm{\Sigma}^{\herm}\!\to\!\bm{\Sigma}$, source independence). A
fixed rule \emph{ordering} is essential because the rules are not
confluent: a Woodbury expansion, for example, can pre-empt an
inverse cancellation that a $d$-separation proof needs. This engine
is what reduces $\bm{\Sigma}_{AB|C}$ to the zero matrix in the
identity quoted above. The second is a \emph{Wirtinger
differentiation engine}, the matrix-calculus analogue, applying
$d\log\det\bm{M}\!=\!\tr(\bm{M}^{-1}d\bm{M})$ and extracting the
coefficient of $d\bm{F}^{\herm}$ to return the closed-form gradient
$\partial\MI/\partial\bm{F}^{\herm}$ for any gain or precoder
$\bm{F}$, as in~\eqref{eq:symbolic_gradient}.

Both engines keep every quantity block-symbolic
(no expansion into scalar entries), so the results stay
dimension-independent. Compound objects such as conditional
covariances and Schur complements are stored as named intermediate
symbols and resolved lazily on demand, so each engine operates on a
compact set of high-level matrix terms rather than on a fully
expanded expression tree. The engines are themselves implemented as
a thin matrix-aware layer over a general-purpose symbolic-computation
system (SymPy in the present implementation), with
\texttt{symbolic-dag} supplying the block-identity rule sets, the
phased ordering, and the Wirtinger calculus discipline that the
underlying CAS lacks.

Every symbolic result is cross-validated against the numerical CMI
evaluator of~\cref{prop:cmi_schur} and its PyTorch autograd, and
\texttt{symbolic-dag} also renders its closed forms as \LaTeX{} for
the analyst. The library is available at
\url{https://github.com/wadayama/symbolic-dag}, together with its
verification suite and documentation.



\begin{thebibliography}{99}

\bibitem{cover2006}
T.~M. Cover and J.~A. Thomas, \emph{Elements of Information Theory},
2nd~ed. Hoboken, NJ: Wiley-Interscience, 2006.

\bibitem{elgamal2011}
A.~El~Gamal and Y.-H. Kim, \emph{Network Information Theory}.
Cambridge, U.K.: Cambridge Univ.\ Press, 2011.

\bibitem{telatar1999}
\.{I}.~E. Telatar,
``Capacity of multi-antenna {Gaussian} channels,''
\emph{Eur.\ Trans.\ Telecommun.}, vol.~10, no.~6, pp.~585--595, Nov.\ 1999.

\bibitem{palomar2006}
D.~P. Palomar and S.~Verd\'u,
``Gradient of mutual information in linear vector {Gaussian} channels,''
\emph{IEEE Trans.\ Inf.\ Theory}, vol.~52, no.~1, pp.~141--154, Jan.\ 2006.

\bibitem{tse1998multiaccess}
D.~N.~C.~Tse and S.~V.~Hanly,
``Multiaccess fading channels---{Part~I}:
Polymatroid structure, optimal resource allocation, and throughput
capacities,''
\emph{IEEE Trans.\ Inf.\ Theory}, vol.~44, no.~7, pp.~2796--2815, Nov.\ 1998.

\bibitem{paszke2019pytorch}
A.~Paszke \emph{et al.},
``{PyTorch}: An imperative style, high-performance deep learning library,''
in \emph{Proc.\ NeurIPS}, 2019, pp.~8024--8035.

\bibitem{hankobayashi1981}
T.~S.~Han and K.~Kobayashi,
``A new achievable rate region for the interference channel,''
\emph{IEEE Trans.\ Inf.\ Theory}, vol.~27, no.~1, pp.~49--60, Jan.\ 1981.

\bibitem{cover1979relay}
T.~M.~Cover and A.~A.~El~Gamal,
``Capacity theorems for the relay channel,''
\emph{IEEE Trans.\ Inf.\ Theory}, vol.~25, no.~5, pp.~572--584, Sep.\ 1979.

\bibitem{caire2003}
G.~Caire and S.~Shamai (Shitz),
``On the achievable throughput of a multiantenna {Gaussian} broadcast channel,''
\emph{IEEE Trans.\ Inf.\ Theory}, vol.~49, no.~7, pp.~1691--1706, Jul.\ 2003.

\bibitem{wadayama2026arxiv}
T.~Wadayama and S.~Na,
``Mutual information optimization via {K}-recursion and
automatic differentiation for linear {Gaussian} wireless networks,''
\emph{arXiv preprint arXiv:2606.06982}, Jun.\ 2026.

\bibitem{wadayama2026dag}
T.~Wadayama,
``Information gradient for directed acyclic graphs: A score-based framework
for end-to-end mutual information maximization,''
\emph{arXiv preprint arXiv:2601.01789}, Jan.\ 2026.

\bibitem{mine2018}
M.~I.~Belghazi \emph{et al.},
``Mutual information neural estimation,''
in \emph{Proc.\ Int.\ Conf.\ Mach.\ Learn.\ (ICML)}, 2018, pp.~531--540.

\bibitem{infonce2018}
A.~van~den~Oord, Y.~Li, and O.~Vinyals,
``Representation learning with contrastive predictive coding,''
\emph{arXiv preprint arXiv:1807.03748}, Jul.\ 2018.

\bibitem{baydin2018ad}
A.~G.~Baydin, B.~A.~Pearlmutter, A.~A.~Radul, and J.~M.~Siskind,
``Automatic differentiation in machine learning: A survey,''
\emph{J.\ Mach.\ Learn.\ Res.}, vol.~18, pp.~1--43, 2018.

\bibitem{schreier_scharf2010}
P.~J.~Schreier and L.~L.~Scharf,
\emph{Statistical Signal Processing of Complex-Valued Data: The Theory of
Improper and Noncircular Signals}.
Cambridge, U.K.: Cambridge Univ.\ Press, 2010.

\bibitem{shachter1989}
R.~D. Shachter and C.~R. Kenley,
``Gaussian influence diagrams,''
\emph{Manage.\ Sci.}, vol.~35, no.~5, pp.~527--550, May 1989.

\bibitem{geiger1994}
D.~Geiger and D.~Heckerman,
``Learning Gaussian networks,''
in \emph{Proc.\ 10th Conf.\ Uncertainty Artif.\ Intell.\ (UAI)}, 1994,
pp.~235--243.

\bibitem{sullivant2010}
S.~Sullivant, K.~Talaska, and J.~Draisma,
``Trek separation for Gaussian graphical models,''
\emph{Ann.\ Statist.}, vol.~38, no.~3, pp.~1665--1685, 2010.

\end{thebibliography}
\end{document}